\documentclass[conference]{IEEEtran}
\usepackage{cite}
\usepackage{amsmath,amssymb,amsfonts}

\usepackage{algpseudocode}
\usepackage{algorithm}
\usepackage{amsthm}
\usepackage{array}
\usepackage[usenames]{color}
\usepackage{epsfig}
\usepackage{epstopdf}
\usepackage{extarrows}
\usepackage{graphicx}
\usepackage{graphics}
\usepackage{stfloats}
\usepackage{subfig}
\usepackage{float}
\usepackage{setspace}
\usepackage{bm}
\usepackage{xparse}
\usepackage{mathrsfs}
\usepackage{booktabs}
\usepackage[font={small}]{caption}
\newenvironment{shrinkeq}[1]{
    \bgroup
    \addtolength\abovedisplayshortskip{#1}
    \addtolength\abovedisplayskip{#1}
    \addtolength\belowdisplayshortskip{#1}
    \addtolength\belowdisplayskip{#1}
}{
    \egroup\ignorespacesafterend
}

\newtheorem{Definition}{Definition}

\newtheorem{Lemma}{Lemma}

\newtheorem{Policy}{Policy}
\newtheorem{Scheme}{Scheme}

\newtheorem{BM}{BM}

\IEEEoverridecommandlockouts

\begin{document}
\title{
    Predictive Resource Allocation in mmWave Systems with Rotation Detection
    \thanks{
        This work was supported by the National Natural Science Foundation of China under Grant 62171213. \\\indent
        Rui Wang is the corresponding author (wang.r@sustech.edu.cn).
    }
}

\author{
Yifei~Sun{\textsuperscript\textdagger}{\textsuperscript\textdaggerdbl},
~Bojie~Lv{\textsuperscript\textdagger},
~Rui~Wang{\textsuperscript\textdagger},
~Haisheng~Tan{\textsuperscript\textsection},
~and~Francis~C.~M.~Lau{\textsuperscript\textdaggerdbl}\\
{\textsuperscript\textdagger}Southern University of Science and Technology, Shenzhen, China \\
{\textsuperscript\textdaggerdbl}The University of Hong Kong, Hong Kong, China \\
{\textsuperscript\textsection}University of Science and Technology of China, Hefei, China
}

\maketitle

\begin{abstract}
    Millimeter wave (MmWave) has been regarded as a promising technology to support high-capacity communications in 5G era. However, its high-layer performance such as latency and packet drop rate in the long term highly depends on resource allocation because mmWave channel suffers significant fluctuation with rotating users due to mmWave sparse channel property and limited field-of-view (FoV) of antenna arrays. In this paper, downlink transmission scheduling considering rotation of user equipments (UE) and limited antenna FoV in an mmWave system is optimized via a novel approximate Markov decision process (MDP) method. Specifically, we consider the joint downlink UE selection and power allocation in a number of frames where future orientations of rotating UEs can be predicted via embedded motion sensors. The problem is formulated as a finite-horizon MDP with non-stationary state transition probabilities. A novel low-complexity solution framework is proposed via one iteration step over a base policy whose average future cost can be predicted with analytical expressions. It is demonstrated by simulations that compared with existing benchmarks, the proposed scheme can schedule the downlink transmission and suppress the packet drop rate efficiently in non-stationary mmWave links.
\end{abstract}

\IEEEpeerreviewmaketitle
\section{Introduction}
Millimeter wave (MmWave) communications is one of the key technologies in 5G and beyond systems \cite{6824752}. In order to achieve desired signal-to-noise ratio (SNR) and high data rate, phased arrays with highly directional beams are adopted to overcome huge pathloss. However, because of the sparse propagation paths in mmWave channel between the base station (BS) and user equipments (UEs), pencil-shaped beams and limited field-of-view (FoV) of antenna arrays, UE mobility especially rotation will cause significant SNR fluctuation \cite{8422671}. Ignorance of such SNR fluctuation may lead to severe high-layer performance degradation such as large latency, buffer overflow and packet drop. Fortunately, exploiting statistical channel model and motion information measured from embedded motion sensors at the UEs makes future SNR fluctuation predictable. This raises a new design issue of joint scheduling in a large time-scale with non-stationary but predictable channel statistics caused by UE rotation.

There has been a number of works considering resource allocation in mmWave MIMO systems either within channel coherent time \cite{8170332} or in a larger time-scale with stationary channel statistics \cite{8782638, 8844783,8536429,7961156}. Scheduling in the former scenario leads to deterministic optimization problems, and infinite-horizon Markov decision process (MDP) or Lyapunov optimization is usually applied for the latter scenario. However, both methods are not applicable for large-time-scale scheduling with non-stationary channel statistics. Scheduling over non-stationary channel statistics with prediction of future channel statistics shall be addressed. For example, the authors in \cite{9184011} proposed an adaptive design for beam alignment, data transmission and handover by exploiting vehicle mobility in mmWave vehicular networks. In \cite{9171304, 9246715}, predictive beamforming is investigated by sensing vehicle mobility with radar to improve the efficiency of resource allocation. However, these works focus on vehicles driving along straight lanes considering only translation while UE rotation may raise more stringent requirement \cite{8422671}. Moreover, current works investigating temporal correlations of angle of departure (AoD) and angle of arrival (AoA) \cite{8686221} or exploiting motion sensors to capture device rotation \cite{9024551,8422671} to assist beam alignment neglect the effect of rapid SNR fluctuation to high-layer performance.

Channel fluctuation due to limited antenna FoV and UE rotation shall be considered in mmWave scheduling. Fortunately, future UE orientation can be predicted according to current angular velocity or angular acceleration measured by motion sensors \cite{6906608,9099976} or integrated sensing and communications (ISAC) techniques \cite{9737357}. Moreover, due to the small-scale fading, it might not be efficient to determine all the transmission parameters in a large time-scale at its beginning as in \cite{9184011, 9171304, 9246715}. Instead, dynamic programming might be a better framework to address the predictive scheduling algorithm design with random channel fading. In this paper, we consider the downlink scheduling with UE rotation where the angular velocity can be measured by motion sensors.

In this paper, we would like to shed some light on the predictive transmission scheduling in mmWave systems with channel fluctuation due to UE rotation and limited antenna FoV. We utilize the motion sensors embedded in UEs to track the orientation of phased arrays. Moreover, by exploiting temporal correlations of mmWave channel, the non-stationary channel statistics can be predicted. We formulate the delay-aware transmission scheduling with non-stationary mmWave channel statistics as a finite-horizon MDP. Finally, a low-complexity approximate MDP solution framework, applying one iteration step over analytically approximated value function, is proposed to reduce the computational complexity. To the best of our knowledge, this is the first paper on mmWave queue-aware scheduling in a large time-scale with consideration of UE rotation.

\section{System Model}
\label{sec:system_model}
\subsection{mmWave System with UE Rotation}
We consider the downlink transmission in an mmWave communication system with one BS and $K$ UEs, where the set of UEs is denoted by $\mathcal{K}\!\triangleq\!\{1,2,\ldots,K\}$. The analog MIMO transceiver with a single radio frequency (RF) chain and a half-wavelength uniform linear phased array (ULA) is adopted at both the BS and UEs. The linear phased arrays at the BS and UEs are with $N_{\mathrm{T}}$ and $N_{\mathrm{R}}$ antenna elements, respectively. Hence, analog precoder and combiners can be adopted at the BS and UEs respectively to enhance the receiving SNR. For elaboration convenience, we consider the mmWave communication in a two-dimensional plane as illustrated in Fig. \ref{fig:system_model}(a).

We focus on mmWave communication scenarios for UEs with rotation but no translation, e.g., playing games with virtual reality (VR) headsets or watching videos when sitting on a rotating chair. Due to UE rotation and limited antenna FoVs, channel statistics such as the number of in-FoV signal scatterers and their directions with respect to the phased arrays will change over time. Fortunately, embedded motion sensors such as magnetometers and gyroscopes are able to periodically detect the orientation and rotation of UEs, respectively \cite{8422671}. Moreover, future UE orientations can be predicted by current angular velocity \cite{6906608}. With the direction knowledge of scatterers and phased arrays, the channel statistics in the future can also be predicted to assist predictive scheduling by jointly considering current and future transmission costs.

Specifically, the transmission time is organized by frames, and the wireless channel is assumed to be quasi-static within one frame. UEs are rotating with predictable angular velocities in a number of frames. The period during which all the UEs are rotating with predictable angular velocities is referred to as one {\textit{scheduling period}} consisting of $T$ frames. For the elaboration convenience, we assume that all UEs are predicted to rotate with constant angular velocity as in \cite{9099976,8277251} in a scheduling period. We shall focus on the joint UE selection and power allocation within one scheduling period.

In the considered scheduling period, the angular velocity of the $k$-th UE is denoted as $\omega_{k}$. The boresight direction of the $k$-th UE in the $1$-st frame of the scheduling period is denoted as $\mathbf{n}_{1,k}$, and the boresight direction of the BS's array is denoted as $\mathbf{n}_{BS}$. Then the rotation angle during the $(t\!-\!1)$ frames and the boresight direction of the $k$-th UE in the $t$-th frame can be predicted as $
    \Delta\phi_{t,k}
    \!\triangleq\!
    (t\!-\!1)\omega_{k}T_{\mathrm{F}}
$ and $
    \mathbf{n}_{t,k}
    \!=\!
    \left[
        \begin{matrix}
            \cos\left(\Delta\phi_{t,k}\right) & \!-\!\sin\left(\Delta\phi_{t,k}\right) \\
            \sin\left(\Delta\phi_{t,k}\right) & \quad\cos\left(\Delta\phi_{t,k}\right)
        \end{matrix}
        \right]
    \mathbf{n}_{1,k}
$, respectively, where $t\!\leq\!T$ and $T_{\mathrm{F}}$ denotes the frame duration.

\begin{figure}[tb]
    \centering
    \subfloat[]{
        \includegraphics[width=0.8\linewidth]{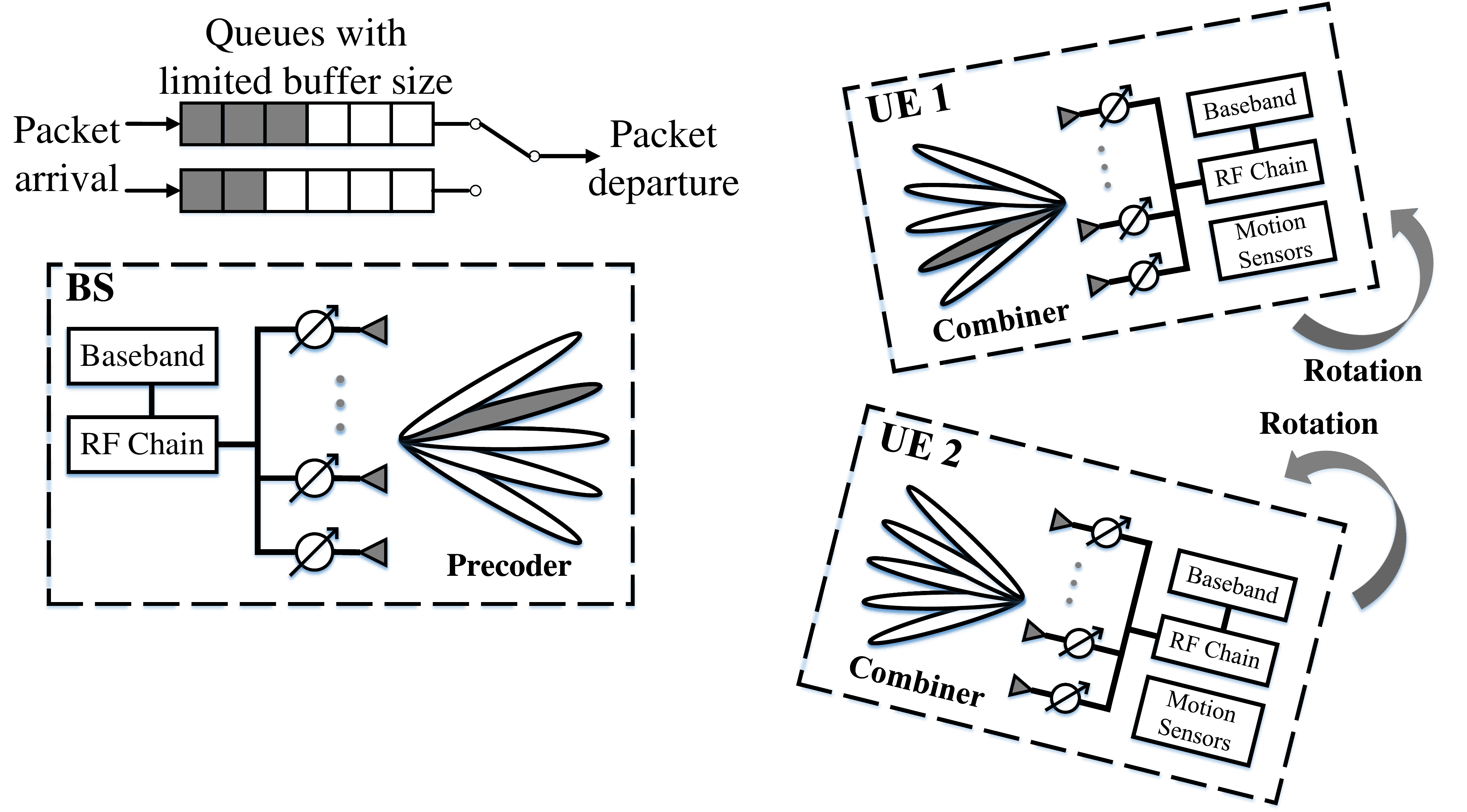}
    }
    \vspace{-0.3cm}
    \hfill
    \subfloat[]{
        \includegraphics[width=0.6\linewidth]{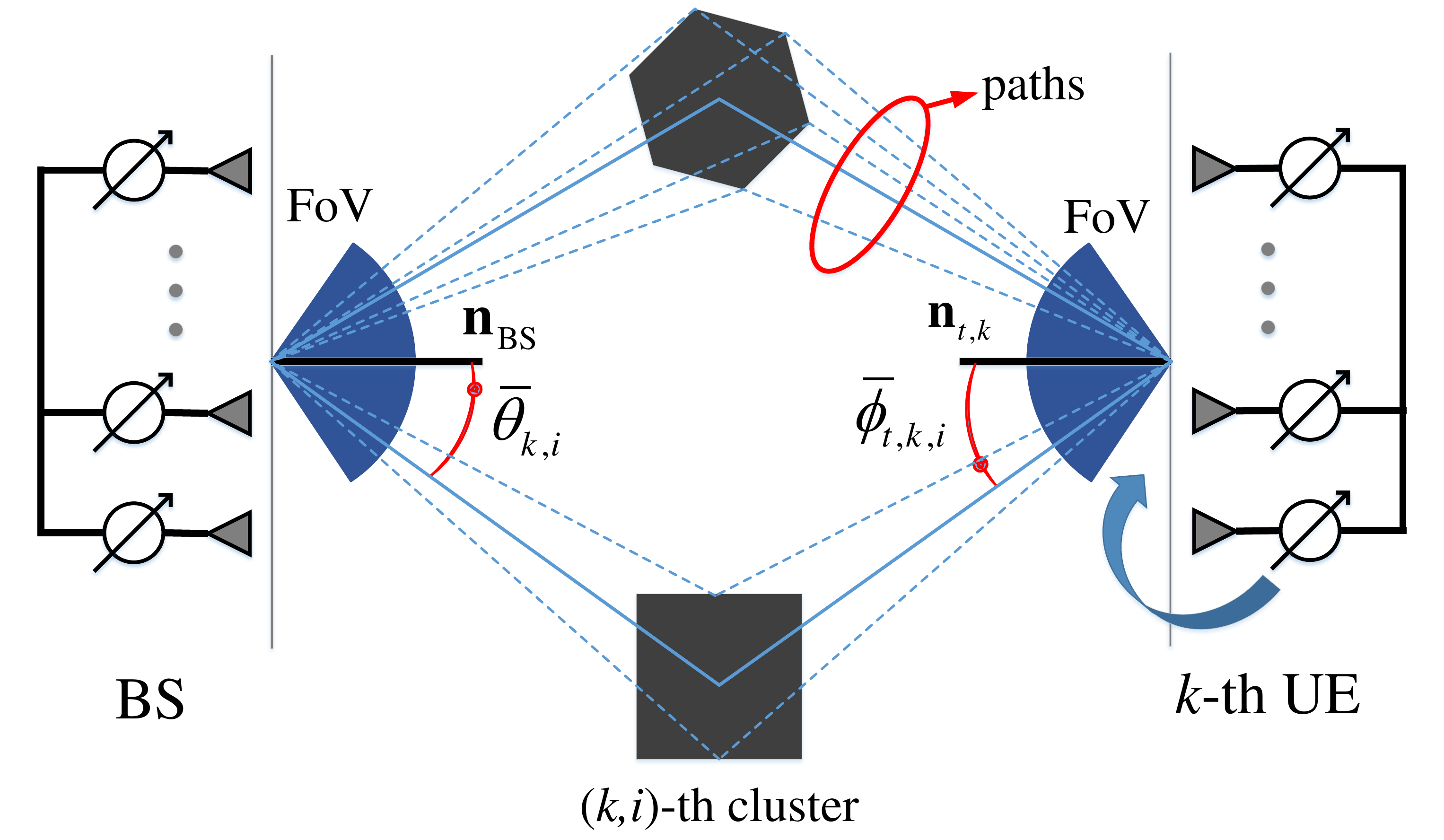}
    }
    \caption{Illustration of (a) system model and (b) channel model.}
    \label{fig:system_model}
    \vspace{-0.5cm}
\end{figure}

\subsection{Channel Model}
\label{subsec:channel_model}
The extended Saleh-Valenzuela channel model illustrated in Fig. \ref{fig:system_model}(b) is adopted. Specifically, there are $N_{k}^{\mathrm{cl}}$ quasi-static scattering clusters in the propagation channel from the BS to the $k$-th UE, and $N_{k,i}^{\mathrm{ray}}$ propagation paths in the $i$-th cluster. For the exposition convenience, the $i$-th cluster in the channel between the BS and the $k$-th UE is referred to as the $(k,i)$-th cluster, and the $\ell$-th path of the $(k,i)$-th cluster is referred to as the $(k,i,\ell)$-th path. Denoting the AoA and AoD of the $(k,i,\ell)$-th path in the $t$-th frame as $\phi_{t,k,i,\ell}$ and $\theta_{t,k,i,\ell}$ respectively, the channel matrix from the BS to the $k$-th UE in the $t$-th frame $
    \mathbf{H}_{t,k}
    \!\in\!
    \mathbb{C}^{N_{\mathrm{R}}\!\times\!N_{\mathrm{T}}}
$ can be represented by
\begin{shrinkeq}{-1ex}
    \begin{align}
        \label{eqn:channel_model}
        \mathbf{H}_{t,k}
        =
         &
        \textstyle{
        \sum_{i=1}^{N_{k}^{\mathrm{cl}}}\sum_{\ell=1}^{N_{k,i}^{\mathrm{ray}}}\alpha_{t,k,i,\ell}\mathbf{a}_{\mathrm{R}}(\phi_{t,k,i,\ell})\mathbf{a}_{\mathrm{T}}^{\mathsf{H}}(\theta_{t,k,i,\ell})
        } \nonumber \\
         &
        \times\Lambda_{\mathrm{R}}(\phi_{t,k,i,\ell})\Lambda_{\mathrm{T}}(\theta_{t,k,i,\ell}),
    \end{align}
\end{shrinkeq}
where $\alpha_{t,k,i,\ell}$ is the instantaneous complex gain of the $(k,i,\ell)$-th path in the $t$-th frame, $\Lambda_{\mathrm{R}}(\phi_{t,k,i,\ell})$ and $\Lambda_{\mathrm{T}}(\theta_{t,k,i,\ell})$ refer to the receiving and transmission antenna gains at $\phi_{t,k,i,\ell} $ and $\theta_{t,k,i,\ell}$ respectively, $\mathbf{a}_{\mathrm{R}}$ and $\mathbf{a}_{\mathrm{T}}$ represent the array response vectors of the ULAs at UEs and the BS, which can be expressed by $
    \mathbf{a}_{\mathrm{R}}(\phi)
    \!=\!
    \frac{1}{\sqrt{N_{\mathrm{R}}}}\left[1,e^{-j\pi\sin(\phi)},\ldots,e^{-j\pi(N_{\mathrm{R}}-1)\sin(\phi)}\right]^{\mathsf{T}},
$ and $
    \mathbf{a}_{\mathrm{T}}(\theta)
    \!=\!
    \frac{1}{\sqrt{N_{\mathrm{T}}}}\left[1,e^{-j\pi\sin(\theta)},\ldots,e^{-j\pi(N_{\mathrm{T}}-1)\sin(\theta)}\right]^{\mathsf{T}}
$, respectively. The antenna patterns are modeled as identical and ideal sectored, hence,
\begin{shrinkeq}{-1ex}
    \begin{align*}
        \label{eqn:element_pattern}
        \Lambda_{\!\mathrm{R}\!}(\phi)
        \!=\!
        \begin{cases}
            1 & \phi\!\in\![\phi_{\mathrm{min}},\phi_{\mathrm{max}}] \\
            0 & \mathrm{otherwise}
        \end{cases}
        ,
        \Lambda_{\!\mathrm{T}\!}(\theta)
        \!=\!
        \begin{cases}
            1 & \theta\!\in\![\theta_{\mathrm{min}},\theta_{\mathrm{max}}] \\
            0 & \mathrm{otherwise}
        \end{cases}
        .
    \end{align*}
\end{shrinkeq}
Due to the limited FoVs of receiving antennas, i.e., $\Lambda_{\mathrm{R}}$, the phased array can not capture all propagation paths in $360^{\circ}$ azimuth. Thus, the in-FoV scattering clusters may vary during UE rotation. For example, the paths via the scattering cluster represented by black square in Fig. \ref{fig:system_model}(b) may be out of the FoV in a few frames when the $k$-th UE is rotating clockwise.

The instantaneous gains $\{\alpha_{t,k,i,\ell}\}$, AoAs $\{\phi_{t,k,i,\ell}\} $ and AoDs $\{\theta_{t,k,i,\ell}\}$ are drawn from independent distributions in each frame. Specifically, $\alpha_{t,k,i,\ell}$ follows the complex Gaussian distribution with zero mean and variance $\sigma_{\alpha;k,i}^{2}$. The angular distribution of propagation paths $\phi_{t,k,i,\ell}$ and $\theta_{k,i,\ell}$ follow the truncated Laplacian distributions with cluster means $
    \bar{\phi}_{t,k,i}
$ and $
    \bar{\theta}_{k,i}
$, and variances $
    \sigma_{\phi;k,i}^{2}
$ and $
    \sigma_{\theta;k,i}^{2}
$ respectively as in \cite{4277071}. Due to UE rotation, the boresight directions of their phased arrays change in each frame. Since the scattering clusters are quasi-static, the AoA means of the $(k,i)$-th cluster in one scheduling period satisfy the following temporal correlation,
\begin{shrinkeq}{-1ex}
    \begin{equation}
        \label{eqn:rotation}
        \bar{\phi}_{t,k,i}
        =
        \bar{\phi}_{1,k,i}
        +
        \Delta\phi_{t,k}, \forall t,k,i.
    \end{equation}
\end{shrinkeq}
As a result, the channel statistics (distribution of $\mathbf{H}_{t,k}$) are non-stationary in one scheduling period, due to the limited FoVs at the UEs and time-varying AoA cluster means $\bar{\phi}_{t,k,i}$. However, with the AoA cluster means and the angular velocity of the $k$-th UE in the $1$-st frame, i.e., $\{\bar{\phi}_{1,k,i}|\forall i\}$ and $\omega_{k}$, the AoA cluster means of all clusters in the $t$-th frame can be predicted according to \eqref{eqn:rotation}.

For the elaboration convenience, we define the statistical channel state information (SCSI) as the tuple of parameters sufficiently characterizing the distribution of channel matrix in \eqref{eqn:channel_model} as follows.
\begin{Definition}
    The SCSI in $t$-th frame is defined by $
        \mathcal{I}_{t}^{\mathrm{SCSI}}
        \!\triangleq\!\\
        (\mathcal{I}_{\mathrm{sta}}^{\mathrm{SCSI}}\!,\!
        \{\bar{\phi}_{t,k,i}|\forall k,i\})
    $, where $
        \mathcal{I}_{\mathrm{sta}}^{\mathrm{SCSI}}
        \!\triangleq\!
        (\{N_{k}^{\mathrm{cl}}|\forall k\}\!,\!
        \{N_{k,i}^{\mathrm{ray}}|\forall k,i\}\!,\!\\
        \{\sigma_{\alpha;k,i}^{2}|\forall k,i\}\!,\!
        \{\bar{\theta}_{k,i}|\forall k,i\}\!,\!
        \{\sigma_{\theta;k,i}^{2}|\forall k,i\}\!,\!
        \{\sigma_{\phi;k,i}^{2}|\forall k,i\})
    $ is the tuple of quasi-static channel statistical parameters.
\end{Definition}

With SCSI in the $1$-st frame $\mathcal{I}_{1}^{\mathrm{SCSI}}$ and $\{\omega_{k}\}$ measured in the $1$-st frame, the SCSI and thus the distribution of $\mathbf{H}_{t,k}$ of all the frames within the scheduling period can be predicted according to the statistical channel model in \eqref{eqn:channel_model}.

\subsection{Sensing-based Beam Alignment}
\label{ssec:beam_alignment}
Due to the single RF chain at the BS, one UE is selected for downlink transmission in each frame. Let $
    \mathbf{w}_{t,k}
    \!\in\!
    \mathbb{C}^{N_{\mathrm{R}}\!\times\!1}
$ and $
    \mathbf{f}_{t,k}
    \!\in\!
    \mathbb{C}^{N_{\mathrm{T}}\!\times\!1}
$ be the analog combiner at the $k$-th UE and the analog precoder at the BS in the $t$-th frame respectively if the $k$-th UE is selected. In practice, $\mathbf{w}_{t,k}$ and $\mathbf{f}_{t,k}$ are selected from the pre-defined finite codebooks $
    \mathbf{w}_{t,k}
    \!\in\!
    \mathcal{W}
    \!\triangleq\!
    \big\{\mathbf{a}_{\mathrm{R}}(\phi_{q})\big|q\!=\!1,2,\ldots,N_{\mathrm{R}}\big\}
$ and $
    \mathbf{f}_{t,k}
    \!\in\!
    \mathcal{F}
    \!\triangleq\!
    \big\{\mathbf{a}_{\mathrm{T}}(\theta_{p})\big|p\!=\!1,2,\ldots,N_{\mathrm{T}}\big\}
$, respectively, where $
    \phi_{q}\!=\!\arcsin\big(\frac{2(q\!-\!1)}{N_{\mathrm{R}}}\!-\!1\big)
$ and $
    \theta_{p}\!=\!\arcsin\big(\frac{2(p\!-\!1)}{N_{\mathrm{T}}}\!-\!1\big)
$. Hence, the spectral efficiency achieved by the $k$-th UE in the $t$-th frame is given by $
    R_{t,k}
    \!=\!
    \log_{2}\big(1
    \!+\!
    \frac{P_{t,k}Y_{t,k}}{N_{0}W}\big)
$, where $P_{t,k}$ is the transmission power of the BS, $
    Y_{t,k}
    \!\triangleq\!
    \big|\mathbf{w}_{t,k}^{\mathsf{H}}\mathbf{H}_{t,k}\mathbf{f}_{t,k}\big|^{2}
$ is the channel power gain in baseband, $N_{0}$ is the noise power spectral density, and $W$ is the bandwidth.

Instead of beamforming with instantaneous CSI feedback which raises significant overhead, we exploit the SCSI prediction and adopt the following statistical beam alignment scheme, which maximizes the average baseband SNR without instantaneous CSI feedback.
\begin{Scheme}[SCSI-Based Beam Alignment]
    \label{scheme:beamforming}
    Given $\mathcal{I}^{\mathrm{SCSI}}_{1}$ and $\{\omega_{k}\}$, the analog combiner and precoder for the $k$-th UE in the $t$-th frame are selected by $
        \mathbf{w}_{t,k}
        \!=\!
        \mathbf{a}_{\mathrm{R}}(\phi_{q_{t}^{\dagger}})
    $ and $
        \mathbf{f}_{t,k}
        \!=\!
        \mathbf{a}_{\mathrm{T}}(\theta_{p_{t}^{\dagger}})
    $, respectively, where $(q_{t}^{\dagger},p_{t}^{\dagger})$ is given by \eqref{eqn:scheme_1}.
\end{Scheme}

\begin{figure*}
    \vspace{0.06in}
    \begin{shrinkeq}{-1ex}
        \begin{align}
            \label{eqn:scheme_1}
            {\textstyle{
            (q_{t}^{\dagger},p_{t}^{\dagger})
            =
            \mathop{\arg\max}\nolimits_{q,p}
            \mathbb{E}_{\mathbf{H}_{t,k}}\left[\left|\mathbf{a}_{\mathrm{R}}^{\mathsf{H}}\left(\phi_{q}\right)\mathbf{H}_{t,k}\mathbf{a}_{\mathrm{T}}\left(\theta_{p}\right)\right|^{2}\right]\qquad\qquad\qquad\qquad\qquad\qquad\qquad\qquad\quad\nonumber
            }}
            \\
            {\textstyle{
            =
            \mathop{\arg\max}\nolimits_{q,p}
            \sum\nolimits_{i=1}^{N_{k}^{\mathrm{cl}}}N_{k,i}^{\mathrm{ray}}\sigma^{2}_{\alpha;k,i}
            \big\{\int_{\phi_{\mathrm{min}}}^{\phi_{\mathrm{max}}}
            \left|\mathbf{a}^{\mathsf{H}}_{\mathrm{R}}\left(\phi_{q}\right)\mathbf{a}_{\mathrm{R}}(\phi_{t,k,i,1})\right|^{2}
            f_{\phi;k,i}(\phi_{t,k,i,1})\mathrm{d}\phi_{t,k,i,1}\big\} \nonumber
            }}
            \\
            {\textstyle{
            \times\big\{\int_{\theta_{\mathrm{min}}}^{\theta_{\mathrm{max}}}
            \left|\mathbf{a}_{\mathrm{T}}^{\mathsf{H}}(\theta_{t,k,i,1})\mathbf{a}_{\mathrm{T}}\left(\theta_{p}\right)\right|^{2}
            f_{\theta;k,i}(\theta_{t,k,i,1})\mathrm{d}\theta_{t,k,i,1}\big\}.\qquad\qquad\qquad\qquad\qquad\qquad
            }}
            \\
            \label{eqn:CDF_Y}
            {\textstyle{
            F_{Y_{t,k}}(x)
            \triangleq\Pr[Y_{t,k}\leq x]=\mathbb{E}_{\{\Phi_{t,k,i}|\forall i\}}\big\{1-\exp(-x/{\textstyle{\sum\nolimits_{i=1}^{N_{k}^{\mathrm{cl}}}}}\Phi_{t,k,i}\sigma_{\alpha;k,i}^{2})\big\}
            ,\ x>0,\qquad\qquad
            }}
            \\
            \label{eqn:binomial}
            {\textstyle{
                    \Phi_{t,k,i}
                    \!\sim\!
                    \mathtt{Binomial}\big(N_{k,i}^{\mathrm{ray}},\Pr\big[\phi_{t,k,i,\ell}\!\in\!\mathcal{P}_{\mathrm{R};q_{t}^{\dagger}}\!\cap\![\phi_{\mathrm{min}},\phi_{\mathrm{max}}]\big]\Pr\big[\theta_{t,k,i,\ell}\!\in\!\mathcal{P}_{\mathrm{T};p_{t}^{\dagger}}\!\cap\![\theta_{\mathrm{min}},\theta_{\mathrm{max}}]\big]\big)
                }}
        \end{align}
    \end{shrinkeq}
    \hrulefill
    \vspace{-1.5\baselineskip}
\end{figure*}

Since the integrals in \eqref{eqn:scheme_1} depends on SCSI in the $t$-th frame $\mathcal{I}_{t}^{\mathrm{SCSI}}$ which can be predicted from $\mathcal{I}_{1}^{\mathrm{SCSI}}$ and $\{\omega_{k}\}$, the precoders and combiners for all the frames within one scheduling period can be pre-designed in the beginning of $1$-st frame. Note that different from the statistical beamforming proposed in \cite{SBF}, we consider non-stationary AoA/AoD cluster means and limited FoVs of ULA. Therefore, beam may switch towards another cluster in advance when previously steered cluster becomes out of FoV due to UE rotation. Moreover, given Scheme \ref{scheme:beamforming}, the cumulative distribution function (CDF) of the baseband channel power gain can be derived as follows.
\begin{Lemma}[CDF of $Y_{t,k}$]
    \label{lem:CDF_Y}
    With Scheme \ref{scheme:beamforming}, when $N_{\mathrm{R}}$ and $N_{\mathrm{T}}$ are sufficiently large, the CDF of $Y_{t,k}$ is given by \eqref{eqn:CDF_Y}, where $\Phi_{t,k,i}$ follows the binomial distribution given by \eqref{eqn:binomial}, and
    \begin{shrinkeq}{-1ex}
        \begin{align}
             & {\textstyle{
            \mathcal{P}_{\mathrm{R};q}
            =\big\{\phi\big|\big|\sin(\phi)-\frac{2(q-1)-N_{\mathrm{R}}}{N_{\mathrm{R}}}\big|\leq\frac{1}{N_{\mathrm{R}}}\big\}
            }},
            \\
             & {\textstyle{
            \mathcal{P}_{\mathrm{T};p}
            =\big\{\theta\big|\big|\sin(\theta)-\frac{2(p-1)-N_{\mathrm{T}}}{N_{\mathrm{T}}}\big|\leq\frac{1}{N_{\mathrm{T}}}\big\}
            }}.
        \end{align}
    \end{shrinkeq}
\end{Lemma}
\begin{proof}
    Please refer to Appendix \ref{proof:CDF_of_Y}.
\end{proof}

\subsection{System Queue Dynamics}
There are $K$ queues for downlink transmission at the BS, each for one UE. The arrival data of each queue is organized by packets, each with $B$ information bits. It is assumed that the number of arrival packets at the $k$-th UE in the $t$-th frame, denoted as $A_{t,k}$, follows independent Poisson distribution with expectation $\lambda_{k}$ across UEs as in \cite{8653854}, i.e., $
    \Pr[A_{t,k}\!=\!n]
    \!=\!
    (\lambda_{k}^{n}/n!)e^{\!-\!\lambda_{k}}
$. Let $
    \mathcal{A}_{t}
    \!\triangleq\!
    \{A_{t,k}|\forall k\!\in\!\mathcal{K}\}
$ represent the aggregated packet arrivals in the $t$-th frame. Without loss of generality, it is assumed that all packets arrive at the end of each frame.

Suppose the $d_{t}$-th UE is selected in the $t$-th frame, the number of departure packets from the $d_{t}$-th queue in the $t$-th frame is given by $
    D_{t,d_{t}}
    \!=\!
    \lfloor WR_{t,d_{t}}T_{\mathrm{F}}/B\rfloor
$. Denote $Q_{t,k}$ as the queue length of the $k$-th queue at the beginning of $t$-th frame and $Q_{\mathrm{max}}$ as the buffer size for each queue both in terms of packets. The queue dynamics can be expressed as $
    Q_{t\!+\!1,k}
    \!=\!
    \min\{Q_{t,k}^{\mathrm{D}}\!+\!A_{t,k},Q_{\mathrm{max}}\}
$, where the arrival packets will be discarded if the buffer is full. The post-decision queue length $Q_{t,k}^{\mathrm{D}}$ is defined as
\begin{shrinkeq}{-1ex}
    \begin{align}
        \label{eqn:Q_tkD}
        Q_{t,k}^{\mathrm{D}}=
        \begin{cases}
            (Q_{t,k}-D_{t,k})^{+} & k=d_{t},     \\
            Q_{t,k}               & k\neq d_{t},
        \end{cases}
    \end{align}
\end{shrinkeq}
where $
    (\cdot)^{+}
    \!\triangleq\!
    \max(0,\cdot)
$.

\section{Problem Formulation}
\label{sec:problem}
Given the precoder and combiner design in Scheme 1, we shall formulate the UE selection and power allocation for all the frames in one scheduling period as a finite-horizon MDP. Note that this is because the system adopts the same prediction for UE rotations in all the frames of one scheduling period. In order to facilitate the MDP formulation, the system state, scheduling action and policy, and post-decision scheduling policy are first elaborated as follows.

\begin{Definition}[System State]
    At the beginning of the $t$-th frame, the system state is represented by $
        \mathcal{S}_{t}
        \!\triangleq\!
        (\mathcal{Q}_{t},\mathcal{Y}_{t})
    $, consisting of queuing state information (QSI) of all the UEs $
        \mathcal{Q}_{t}
        \!\triangleq\!
        \{Q_{t,1},Q_{t,2},\ldots,Q_{t,K}\}
    $, and baseband channel power gains to all the UEs $
        \mathcal{Y}_{t}
        \!\triangleq\!
        \{Y_{t,1},Y_{t,2},\ldots,Y_{t,K}\}$.
\end{Definition}

\begin{Definition}[Scheduling Action and Policy]
    At the beginning of the $t$-th frame, the scheduling actions include the UE selection for downlink transmission $
        d_{t}
        \!\in\!
        \mathcal{K}
    $ and the downlink transmission power $
        P_{t}
        \!\triangleq\!
        P_{t,d_{t}}
    $, where the following instantaneous power constraint at the BS should be satisfied,
    \begin{shrinkeq}{-1ex}
        \begin{equation}
            \label{constraint:power}
            P_{t}
            \leq
            P_{\mathrm{max}},
            \ \forall t.
        \end{equation}
    \end{shrinkeq}
    Hence, the scheduling policy of the BS, denoted as $\Omega_{t}$, is a mapping from the system state $\mathcal{S}_{t}$ to the scheduling actions. Thus, $
        \Omega_{t}(\mathcal{S}_{t})
        \!\triangleq\!
        (d_{t},P_{t})
    $.
\end{Definition}

\begin{Definition}[Post-Decision State]
    \label{def:Q_tD}
    At the beginning of the $t$-th frame, the post-decision system state is defined by $\mathcal{S}_{t}^{\mathrm{D}}\triangleq (\mathcal{Q}_{t}^{\mathrm{D}},\mathcal{Y}_{t})$, where $
        \mathcal{Q}_{t}^{\mathrm{D}}
        \!\triangleq\!
        \{Q_{t,1}^{\mathrm{D}},Q_{t,2}^{\mathrm{D}},\ldots,Q_{t,K}^{\mathrm{D}}\}
    $ denotes post-decision QSI of all the UEs.
\end{Definition}

The \textit{post-decision system state} is the system state after packet transmission but before packet arrivals. Given the above definition of post-decision system state and scheduling policy, the transition probability is given by $
    \Pr(\mathcal{S}_{t\!+\!1}^{\mathrm{D}}|\mathcal{S}_{t}^{\mathrm{D}},\Omega_{t\!+\!1})
    \!=\!
    \Pr(\mathcal{Q}_{t\!+\!1}^{\mathrm{D}}|\mathcal{Q}_{t}^{\mathrm{D}},\Omega_{t\!+\!1})
    \Pr(\mathcal{Y}_{t\!+\!1})
    \!=\!
    \prod_{k\!=\!1}^{K}\Pr(Q_{t\!+\!1,k}^{\mathrm{D}}|Q_{t,k}^{\mathrm{D}},\Omega_{t\!+\!1,k})
    \prod_{k\!=\!1}^{K}\Pr(Y_{t\!+\!1,k})
$.

In this paper, the scheduling policies are designed to optimize the system queuing performance, while saving the average total energy consumption. Specifically, the transmission energy, average packet transmission delay and penalty of packet drop are considered as the costs of scheduling. According to Little's law, the sum of queuing packet numbers of all the frames in the scheduling period can be used as an equivalent measurement of average transmission delay. Hence, we define the following weighted sum of the transmission power consumption, the number of queuing packets for all UE and full buffer penalty as the system cost in the $t$-th frame,
\begin{shrinkeq}{-1ex}
    \begin{align*}
        {\textstyle{
        g_{t}\big(\mathcal{S}_{t},\Omega_{t}(\mathcal{S}_{t})\big)
        \!\triangleq\!
        w_{\mathrm{P}}P_{t}
        \!+\!
        \sum_{k\in\mathcal{K}}\big(Q_{t,k}
        \!+\!
        w_{\mathrm{Q}}\mathbb{I}[Q_{t,k}\!=\!{Q}_{\mathrm{max}}]\big),
        }}
    \end{align*}
\end{shrinkeq}
where $w_{\mathrm{P}}$ and $w_{\mathrm{Q}}$ are the weights of power consumption and full buffer penalty respectively, and $\mathbb{I}[\cdot]$ denotes an indicator function, which is $1$ when the event is true and $0$ otherwise.

The overall minimization objective of one scheduling period with the initial system state $\mathcal{S}_{1}$ is then given by
\begin{shrinkeq}{-1ex}
    \begin{align*}
        \textstyle{
        G(\mathcal{S}_{1},\Omega)
        \!\triangleq\!
        \mathbb{E}_{\mathcal{A},\mathcal{Y}}^{\Omega}\big[\sum_{t\!=\!1}^{T}g_{t}\big(\mathcal{S}_{t},\Omega_{t}(\mathcal{S}_{t})\big)
        \!+\!
        {\varrho}\big(\mathcal{Q}_{T+1}\big)\big|\mathcal{S}_{1}\big],
        }
    \end{align*}
\end{shrinkeq}
where $
    \mathcal{A}
    \!\triangleq\!
    \{\mathcal{A}_{t}|\forall t\}
$, $
    \mathcal{Y}
    \!\triangleq\!
    \{\mathcal{Y}_{t}|\forall t\}
$, $
    \Omega
    \!\triangleq\!
    \{\Omega_{t}|\forall t\}
$, and $
    \varrho(\mathcal{Q}_{T\!+\!1})
    \!\triangleq\!
    \sum_{k\in\mathcal{K}}Q_{T\!+\!1,k}
$ counts for the remaining packet number at the end of one scheduling period. The expectation is taken with respect to the randomness of packet arrivals $\mathcal{A}$ and baseband channel power gains $\mathcal{Y}$. As a result, the transmission design in this paper can be formulated as the following dynamic programming problem.
\begin{shrinkeq}{-1ex}
    \begin{align*}
        \bm{\mathsf{P1:}}\ \Omega^{\star}
        \triangleq
        \{\Omega_{t}^{\star}|\forall t\}
        =
         &
        \textstyle{
            \mathop{\arg\min}_{\Omega}G(\mathcal{S}_{1},\Omega)
        }  \\
         &
        \mathrm{s.t.}\quad 0\leq P_{t}\leq P_{\mathrm{max}},\ \forall t.
    \end{align*}
\end{shrinkeq}

We shall adopt the Bellman's equations with {\textit{post-decision value functions}} to solve $\mathsf{P1}$. Note that the baseband channel power gains are independently distributed in different frames, it can be averaged as in \cite{8653854}, and the Bellman's equation based on the post-decision QSI only can be written as
\begin{shrinkeq}{-1ex}
    \begin{align}
        \label{eqn:bellman1}
        W_{t}(\mathcal{Q}_{t}^{\mathrm{D}})
        =
         &
        \textstyle{
        \min_{\Omega_{t+1}(\mathcal{S}_{t+1})}\mathbb{E}_{\mathcal{A}_{t},\mathcal{Y}_{t+1}}\big[g_{t+1}\big(\mathcal{S}_{t+1},\Omega_{t+1}(\mathcal{S}_{t+1})\big)
        }\nonumber \\
         &
        \!+\!
        W_{t+1}\big(\mathcal{Q}_{t+1}^{\mathrm{D}}\big)\big|\mathcal{Q}_{t}^{\mathrm{D}}\big]
    \end{align}
\end{shrinkeq}
where $
    W_{t}(\mathcal{Q}_{t}^{\mathrm{D}})
$ is the post-decision value function of the optimal policy (referred to as optimal value function for short), and $
    W_{T}(\mathcal{Q}_{T}^{\mathrm{D}})
    \!\triangleq\!
    \mathbb{E}_{\mathcal{A}_{T}}\varrho(\mathcal{Q}_{T+1}|\mathcal{Q}_{T}^{\mathrm{D}})
$ for the notation convenience. The Bellman's equations with post-decision system state can avoid complicated calculation of transition matrix and expectation of value functions\cite{powell2007approximate}. Moreover, the optimal scheduling policy in the $t$-th frame $(\forall t\!=\!1,2,\ldots,T)$ for $\mathsf{P1}$ can be obtained by
\begin{shrinkeq}{-1ex}
    \begin{align}
        \label{eqn:optimal-policy}
        \Omega_{t}^{\star}(\mathcal{S}_{t})
        =
         &
        \textstyle{
        \mathop{\arg\min}_{\Omega_{t}(\mathcal{S}_{t})}
        \big\{g_{t}\big(\mathcal{S}_{t},\Omega_{t}(\mathcal{S}_{t})\big)
        +
        W_{t}\big(\mathcal{Q}_{t}^{\mathrm{D}}\big)\big\}
        }                                                                            \\
         & \mathrm{s.t.}\quad 0\leq P_{t}\leq P_{\mathrm{max}},\ \forall t.\nonumber
    \end{align}
\end{shrinkeq}

\section{Low-Complexity Scheduling}
\label{sec:low-complexity_scheduling}
It can be observed from \eqref{eqn:optimal-policy} that the optimal value function for the $t$-th frame $W_{t}(\mathcal{Q}_{t}^{\mathrm{D}})$ should be calculated before the derivation of the optimal scheduling policy for the $t$-th frame $\Omega_{t}(\mathcal{S}_{t})$. However, because of the minimization in \eqref{eqn:optimal-policy}, it is difficult to derive the closed-form expression for the optimal value function in each frame $\{W_{t}(\mathcal{Q}_{t}^{\mathrm{D}})|\forall t\}$. It is also difficult to calculate the value functions for all possible system states numerically due to the huge system state space. In this section, a low-complexity solution framework is proposed. Specifically, we first adopt the backpressure algorithm \cite{georgiadis2006resource} as the base policy and derive the closed-form expressions of its value functions (referred to as approximate value function for short). Then the approximate value functions are used to approximate the optimal value functions to obtain the improved scheduling policy by one-step policy improvement.

\subsection{Base Policy}
As in \eqref{eqn:bellman1}, the base policy provides an approximation of average future cost to improve current scheduling actions. It should have a good scheduling performance and a simple structure for analysis. Hence, as the base policy, we predetermine the transmission powers and UE selection for all frames at the very beginning of one scheduling period. Particularly, given the system QSI in the $1$-st frame, we use the average downlink throughput and packet arrival rate to approximate the queue dynamics, and adopt the {\textit{backpressure algorithm}} \cite{georgiadis2006resource} to select the downlink UE. Let $P_{\mathrm{BSL}}$ be the predetermined transmission power and $Q^{\dagger}_{t,k}$ be the approximate queue length of the $k$-th queue at the beginning of the $t$-th frame, then the index of the selected UE is given by
\begin{shrinkeq}{-1ex}
    \begin{align}
        \label{eqn:UE-selection}
        \textstyle{
        d_{t}^{\Pi}
        \!=\!
        \mathop{\arg\max}_{k\in\mathcal{K}}\big[R_{t,k}^{\dagger}(P_{\mathrm{BSL}})Q_{t,k}^{\dagger}\big],\ \forall t,
        }
    \end{align}
\end{shrinkeq}
where $R_{t,k}^{\dagger}$ is the predicted average spectral efficiency of the $k$-th UE in the $t$-th frame, i.e., $
    \textstyle{
    R_{t,k}^{\dagger}(P_{\mathrm{BSL}})
    \!\triangleq\!
    \int\log_{2}\big(1\!+\!\frac{P_{\mathrm{BSL}}x}{N_{0}W}\big)
    \frac{\mathrm{d}F_{Y_{t,k}}(x)}{\mathrm{d}x}\mathrm{d}x,
    \ \forall t,k\!\in\!\mathcal{K}
    }
$. Moreover, given UE selection in the $t$-th frame, the QSI in the $(t\!+\!1)$-th frame can be approximated by
\begin{shrinkeq}{-1ex}
    \begin{align}
        \label{eqn:approximate_QSI}
        Q_{t+1,k}^{\dagger}
        \!=\!
        \begin{cases}
            \min\{(Q_{t,k}^{\dagger}\!-\!D_{t,k}^{\dagger})^{+}\!+\!\lambda_{k},Q_{\mathrm{max}}\}
             & k\!=\!d_{t}^{\Pi},    \\
            \min\{Q_{t,k}^{\dagger}\!+\!\lambda_{k},Q_{\mathrm{max}}\}
             & k\!\neq\!d_{t}^{\Pi},
        \end{cases}
    \end{align}
\end{shrinkeq}
where $
    Q_{1,k}^{\dagger}
    \!=\!
    Q_{1,k}
$ and $
    D_{t,k}^{\dagger}
    \!\triangleq\!
    \big\lfloor WR_{t,k}^{\dagger}(P_{\mathrm{BSL}})N_{\mathrm{F}}/B\big\rfloor,
    \allowbreak\ \forall t,k\!\in\!\mathcal{K}
$. By applying \eqref{eqn:UE-selection} and \eqref{eqn:approximate_QSI} iteratively, the UE selection of the base policy can be determined. Hence, the base policy can be summarized as following.

\begin{Policy}[Base Policy $\Pi$]
    The transmission power to each selected UE for base policy is fixed, i.e., $P_{t}\!=\!P_{\mathrm{BSL}},\forall t$. Moreover, the UE selection is determined by applying \eqref{eqn:UE-selection} and \eqref{eqn:approximate_QSI} iteratively, which is denoted as $\{d_{1}^{\Pi},d_{2}^{\Pi},\ldots,d_{T}^{\Pi}\}$.
\end{Policy}

Although the base policy is obtained by approximating queue dynamics with averaging, the approximate value function should be evaluated with actual distribution of packet arrivals and departures. Hence, the approximate value function, which approximately measures the average system cost for all the UEs from the $(t\!+\!1)$-th frame, can be written as
\begin{align}
    \label{eqn:W_pi_decouple}
    \textstyle{
    W_{t}^{\Pi}(\mathcal{Q}_{t}^{\mathrm{D}})
    =
    (T-t)w_{\mathrm{P}}P_{\mathrm{BSL}}
    +
    \sum_{k\in\mathcal{K}}W_{t,k}^{\Pi}({Q}_{t,k}^{\mathrm{D}}),
    }
\end{align}
where the {\textit{local value function}} $W_{t,k}^{\Pi}({Q}_{t,k}^{\mathrm{D}})$ is defined as the average queuing cost raised by the $k$-th UE from the $(t\!+\!1)$-th frame to the end of the current scheduling period given the base policy $\Pi$ and its {\textit{local post-decision QSI}} ${Q}_{t,k}^{\mathrm{D}}$. Thus,
\begin{shrinkeq}{-1ex}
    \begin{align}
        W_{t,k}^{\Pi}({Q}_{t,k}^{\mathrm{D}})
        \triangleq
         & \mathbb{E}_{\mathcal{A},\mathcal{Y}}^{\Pi}\big[{\textstyle{\sum_{\tau=t+1}^{T}}}\big(Q_{\tau,k}\!+\!w_{\mathrm{Q}}\mathbb{I}[Q_{\tau,k}\!=\!{Q}_{\mathrm{max}}]\big) \nonumber \\
         & \!+\!
        Q_{T+1,k}\big|Q_{t,k}^{\mathrm{D}}\big],\ \forall{k}\in\mathcal{K}.
    \end{align}
\end{shrinkeq}

In order to derive the analytical expression of $W_{t,k}^{\Pi}({Q}_{t,k}^{\mathrm{D}})$, denote the number of departure packets from the $k$-th queue in the $t$-th frame under the base policy as $D_{t,k}^{\Pi}$ if the $k$-th UE is selected, then we have the following lemma.
\begin{Lemma}
    \label{lem:PDF_of_D}
    With Scheme \ref{scheme:beamforming}, the probability mass function (PMF) of $D_{t,k}^{\Pi}$ is given by $
        \Pr[D_{t,k}^{\Pi}\!=\!n]
        \!=\!
        \textstyle{
            F_{Y_{t,k}}\big[\big(2^{\frac{(n\!+\!1)B}{WT_{\mathrm{F}}}}\!-\!1\big)\frac{N_{0}W}{P_{\mathrm{BSL}}}\big]
            \!-\!
            F_{Y_{t,k}}\big[\big(2^{\frac{nB}{WT_{\mathrm{F}}}}\!-\!1\big)\frac{N_{0}W}{P_{\mathrm{BSL}}}\big]
        }
    $.
\end{Lemma}

\begin{proof}
    The proof is straightforward based on the definition of $D_{t,k}$ and $R_{t,k}$.
\end{proof}

Hence, the local approximate value function $W_{t,k}^{\Pi}({Q}_{t,k}^{\mathrm{D}})$ ($\forall t,k$) is derived in the following lemma.
\begin{Lemma}[Analytical Expression of $W_{t,k}^{\Pi}({Q}_{t,k}^{\mathrm{D}})$]
    \label{lem:analytical_expression}
    Let $
        \mathbf{u}_{t,k},\mathbf{s}_{t,\tau,k},\mathbf{c}^{(\!1\!)},\mathbf{c}^{(\!2\!)}\!\in\!\mathbb{R}^{(Q_{\mathrm{max}}\!+\!1)\!\times\!{1}}
    $, which are defined by $
        \mathbf{u}_{t,k}
        \!\triangleq\!
        \mathbf{1}_{Q_{t,k}^{\mathrm{D}}\!+\!1}
    $, $
        \mathbf{s}_{t,\tau,k}
        \!\triangleq\!
        \mathbf{1}_{Q_{\tau,k}\!+\!1}
    $, $
        [\mathbf{c}^{(\!1\!)}]_{i}
        \!\triangleq\!
        i\!-\!1\!+\!w_{\mathrm{Q}}\mathbb{I}[i\!=\!Q_{\mathrm{max}}\!+\!1]
    $, and $
        [\mathbf{c}^{(\!2\!)}]_{i}
        \!\triangleq\!
        i\!-\!1
    $, respectively. $\mathbf{1}_{i}$ denotes the column vector whose $i$-th element is $1$ and other elements are $0$. Let $
        \mathbf{P}_{k},\mathbf{M}_{t,k}
        \!\in\!
        \mathbb{R}^{(Q_{\mathrm{max}}\!+\!1)\!\times\!(Q_{\mathrm{max}}\!+\!1)}
    $, whose entries are given by \eqref{eqn:P_k} and \eqref{eqn:M_tk}, respectively. Then $W_{t,k}^{\Pi}(Q_{t,k}^{\mathrm{D}})$ can be represented by
    \begin{shrinkeq}{-1ex}
        \begin{align}
            \label{eqn:local-value-analytical}
            \textstyle{
            W_{t,k}^{\Pi}({Q}_{t,k}^{\mathrm{D}})
            =
            \sum_{\tau=t+1}^{T}\mathbf{s}_{t,\tau,k}^{\mathsf{T}}\mathbf{c}^{(\!1\!)}
            +
            \mathbf{s}_{t,T+1,k}^{\mathsf{T}}\mathbf{c}^{(\!2\!)},
            }
        \end{align}
    \end{shrinkeq}
    where $
        \mathbf{s}_{t,\tau,k}
        \!=\!
        \mathbf{X}_{k}^{\mathsf{T}}(t,\tau)\mathbf{P}_{k}^{\mathsf{T}}\mathbf{u}_{t,k}
    $, and $
        \mathbf{X}_{k}(t,\tau)
        \!=\!
        \prod_{n=t+1}^{\tau-1}\mathbf{M}_{n,k}^{\mathbb{I}(d_{n}^{\Pi}=k)}\mathbf{P}_{k}^{\mathbb{I}(d_{n}^{\Pi}\neq{k})}
    $, $
        \tau=t+1,\ldots,T
    $.
\end{Lemma}

\begin{proof}
    Please refer to Appendix \ref{proof:analytical_expression}.
\end{proof}

\begin{figure*}
    \vspace{0.05in}
    \begin{shrinkeq}{-1ex}
        \begin{align}
            \label{eqn:P_k}
            [\mathbf{P}_{k}]_{i,j}
            =
            \begin{cases}
                \Pr[A_{t,k}=j-i]                      & 1\leq i\leq Q_{\mathrm{max}} \mbox{ and } i\leq j\leq Q_{\mathrm{max}}, \\
                \Pr[A_{t,k}\geq Q_{\mathrm{max}}+1-i] & 1\leq i\leq Q_{\mathrm{max}} \mbox{ and } j=Q_{\mathrm{max}}+1,         \\
                1                                     & i=Q_{\mathrm{max}}+1 \mbox{ and } j=Q_{\mathrm{max}}+1,                 \\
                0                                     & \mbox{otherwise}
            \end{cases}\quad\quad\quad\quad\quad\quad\quad\quad \\
            \label{eqn:M_tk}
            [\mathbf{M}_{t,k}]_{i,j}
            =
            \begin{cases}
                \Pr[D_{t,k}^{\Pi}\geq i-1]\Pr[A_{t,k}=0]                      & 1\leq i\leq Q_{\mathrm{max}}+1 \mbox{ and } j=1,                          \\
                \Pr[A_{t,k}-\min(D_{t,k}^{\Pi},i-1)=j-i]                      & 1\leq i\leq Q_{\mathrm{max}}+1 \mbox{ and } 2\leq j\leq Q_{\mathrm{max}}, \\
                \Pr[A_{t,k}-\min(D_{t,k}^{\Pi},i-1)\geq Q_{\mathrm{max}}+1-i] & 1\leq i\leq Q_{\mathrm{max}}+1 \mbox{ and } j=Q_{\mathrm{max}}+1
            \end{cases}\quad\quad\quad\quad
        \end{align}
    \end{shrinkeq}
    \hrulefill
    \vspace{-1.5\baselineskip}
\end{figure*}

\subsection{Scheduling with Approximate Value Function}
In this part, we use approximate value functions $\{W_{t}^{\Pi}(\mathcal{Q}_{t}^{\mathrm{D}})|\forall t,\mathcal{Q}_{t}^{\mathrm{D}}\}$ to approximate optimal value functions $\{W_{t}(\mathcal{Q}_{t}^{\mathrm{D}})|\forall t,\mathcal{Q}_{t}^{\mathrm{D}}\}$. Because the approximate value function is analytically expressed, conventional value iteration to evaluate the value function can be avoided, which can significantly reduce the computational complexity. With the approximate value function, the scheduling actions in the $t$-th ($\forall t$) frame could be obtained by solving the following problem.
\begin{shrinkeq}{-1ex}
    \begin{align*}
        \bm{\mathsf{P2}}
         &
        \ \textrm{(One-Step Policy Iteration): }
        \\
         &
        \Psi_{t}(\mathcal{S}_{t})
        \triangleq
        (d_{t}^{\Psi},P_{t}^{\Psi})
        \\
        =
         &
        \textstyle{
        \mathop{\arg\min}_{\Omega_{t}(\mathcal{S}_{t})}\big\{g_{t}\big(\mathcal{S}_{t},\Omega_{t}(\mathcal{S}_{t})\big)
        \!+\!
        W_{t}^{\Pi}\big(\mathcal{Q}_{t}^{\mathrm{D}}(\mathcal{S}_{t},\Omega_{t})\big)\big\},
        }
        \\
         &
        \mathrm{s.t.}\quad 0\leq P_{t}\leq P_{\mathrm{max}}.
    \end{align*}
\end{shrinkeq}

Since both the current system cost and approximate value function in the $t$-th frame can be decomposed by each UE, $\mathsf{P2}$ can be decomposed into $K$ sub-problems. The $k$-th sub-problem is given by
\begin{shrinkeq}{-1ex}
    \begin{align*}
        \bm{\mathsf{P2(k):}}\ P_{t,k}^{\psi}
        =
         &
        \textstyle{
        \mathop{\arg\min}_{P_{t,k}}\big\{w_{p}P_{t,k}
        +
        W_{t}^{\Pi}\big(\mathcal{Q}_{t}^{\mathrm{D}}(\mathcal{S}_{t},\pi_{t,k})\big)\big\},
        }                                                         \\
         & \mathrm{s.t.}\quad 0\leq P_{t,k}\leq P_{\mathrm{max}},
    \end{align*}
\end{shrinkeq}
where $
    \pi_{t,k}
    \!\triangleq\!
    (k,P_{t,k})
$ denotes the policy that the $k$-th UE is selected. Moreover, denote the minimized objective of $\mathsf{P2(k)}$ as $G_{t,k}^{\psi}$, then the optimal solution of $\mathsf{P2}$ can be derived by $
    d_{t}^{\Psi}
    \!=\!
    \mathop{\arg\min}_{k\in\mathcal{K}}G_{t,k}^{\psi}
$ and $
    P_{t}^{\Psi}
    \!=\!
    P_{t,d_{t}^{\Psi}}^{\psi}
$.

The transmission power allocation for $\mathsf{P2(k)}$, which is a discrete optimization problem, can be achieved via the following one-dimensional search.
\begin{Lemma}[Local Power Optimization]
    \label{lem:power}
    The optimized transmission power for the $k$-th UE (solution of $\mathsf{P2(k)}$) is
    \begin{shrinkeq}{-1ex}
        \begin{align}
            \label{eqn:power}
            \textstyle{
            P_{t,k}^{\psi}
            =
            \mathop{\arg\min}_{P_{t,k}\in\mathscr{P}_{t,k}}\big\{w_{\mathrm{P}}P_{t,k}
            +
            \Delta\mathbf{z}_{t,k}^{\mathsf{T}}(P_{t,k})\mathbf{V}_{t,k}\big\},
            }
        \end{align}
    \end{shrinkeq}
    where $
        \textstyle{
            \mathscr{P}_{t,k}
            \!\triangleq\!
            \big\{
            0
            ,
            \!\frac{2^{\frac{B}{WN_{\mathrm{F}}}}\!-\!1}{Y_{t,k}}
            \!,\!
            \frac{2^{\frac{2B}{WN_{\mathrm{F}}}}\!-\!1}{Y_{t,k}}
            \!,\!
            \ldots
            \!,\!
            \min\!\big(\!\frac{2^{\frac{Q_{\mathrm{max}}B}{WN_{\mathrm{F}}}}\!-\!1}{Y_{t,k}}\!,\!P_{\mathrm{max}}\!\big)\!\big\}
        }
    $ is the feasible power set, $
        \Delta\mathbf{z}_{t,k}(P_{t,k})
        \!=\!
        \mathbf{1}_{Q_{t,k}^{\mathrm{D}}(P_{t,k})\!+\!1}
        \!-\!
        \mathbf{1}_{Q_{t,k}\!+\!1}
    $, $
        \textstyle{
            Q_{t,k}^{\mathrm{D}}(P_{t,k})
            \!=\!
            \big(Q_{t,k}
            \!-\!
            \big\lfloor\frac{WR_{t,k}(P_{t,k})T_{\mathrm{F}}}{B}\big\rfloor\big)^{+}
        }
    $, $
        \mathbf{V}_{t,k}
        \!=\!
        \sum_{\tau=t+1}^{T}\mathbf{Y}_{\tau,k}^{\mathsf{T}}\mathbf{c}^{(\!1\!)}
        \!+\!
        \sum_{\tau=t+2}^{T+1}\mathbf{Y}_{\tau,k}^{\mathsf{T}}\mathbf{c}^{(\!2\!)}
    $, and
    \begin{shrinkeq}{-1ex}
        \begin{align}
            \mathbf{Y}_{\tau,k}=
            \begin{cases}
                \mathbf{P}_{k}^{\mathsf{T}}
                 & \tau=t+1,          \\
                \mathbf{X}_{k}^{\mathsf{T}}(t+1,\tau)\mathbf{Y}_{t+1,k}
                 & \tau=t+2,\ldots,T.
            \end{cases}
        \end{align}
    \end{shrinkeq}
\end{Lemma}

\begin{proof}
    The feasible power set $\mathscr{P}_{t,k}$ is the minimum required transmission power to transmit an integer number of packets, thus it will not affect the optimality. Then $\mathsf{P2(k)}$ can be solved by one-dimensional search in $\mathscr{P}_{t,k}$. $
        \Delta\mathbf{z}_{t,k}^{\mathsf{T}}(P_{t,k})\mathbf{V}_{t,k}
        \!=\!
        W_{t,k}^{\Pi}\big(Q_{t,k}^{\mathrm{D}}(P_{t,k})\big)
        \!-\!
        W_{t,k}^{\Pi}(Q_{t,k})
    $, where $W_{t,k}^{\Pi}(Q_{t,k})$ is a constant in $\mathsf{P2(k)}$. Thus, the proof is straightforward.
\end{proof}

As a summary, the system can be scheduled in a semi-distributed manner as elaborated below.
\begin{itemize}
    \item
          \textbf{Step 1:} At the beginning of one scheduling period, each UE (say the $k$-th UE) calculates $\{\mathbf{V}_{t,k}|\forall t\}$ locally based on SCSI $\mathcal{I}^{\mathrm{SCSI}}_{t}$ and the angular velocity $\omega_{k}$.
    \item
          \textbf{Step 2:} At the beginning of the $t$-th frame ($\forall t$), each UE (say the $k$-th UE) calculates the optimal power $P_{t,k}^{\psi}$ by assuming it is selected, and reports $G_{t,k}^{\psi}$ and $P_{t,k}^{\psi}$ to the BS via uplink signaling channel.
    \item
          \textbf{Step 3:} After receiving the reports from all UEs, the BS determines the improved scheduling policy $\Psi_{t}$.
\end{itemize}

\section{Simulations and Discussion}
\label{sec:simulations}
In this part, the performance of the proposed algorithm is demonstrated via numerical simulations with existing benchmark algorithms for comparison. There are eight UEs in the system where four UEs are static with zero angular velocities (indexed with $1\!\sim\!4$) while the other four UEs are rotating with angular velocity $2$ rad/s (indexed with $5\!\sim\!8$). We compare the proposed algorithms with the following three benchmarks, which are referred to BM1, BM2 and BM3. For fair comparison, the fixed transmission power of the base policy and benchmarks are the same, i.e., $P_{\mathrm{BSL}}\!=\!P_{\mathrm{BM}}$.

\begin{BM}[Dynamic Backpressure]
    The transmission power to each selected UE is fixed to $P_{\mathrm{BM}}$. The UE selection is based on the backpressure algorithm \cite{georgiadis2006resource} according to the real-time data rate and queue length, i.e., $
        d_{t}
        \!=\!
        \mathop{\arg\max}_{k\in\mathcal{K}}R_{t,k}(P_{\mathrm{BM}})Q_{t,k}
    $.
\end{BM}

\begin{BM}[Largest-Rate First]
    The transmission power to each selected UE is fixed to $P_{\mathrm{BM}}$. In each frame, the UE with the largest data rate is selected, i.e., $
        d_{t}
        \!=\!
        \mathop{\arg\max}_{k\in\mathcal{K}}R_{t,k}(P_{\mathrm{BM}})
    $.
\end{BM}

\begin{BM}[Longest-Queue First]
    The transmission power to each selected UE is fixed to $P_{\mathrm{BM}}$. In each frame, the UE with the longest queue is selected, i.e., $
        d_{t}
        \!=\!
        \mathop{\arg\max}_{k\in\mathcal{K}}Q_{t,k}
    $.
\end{BM}

The instantaneous SNR and queue length of a static UE (indexed by $k\!=\!1$) and a rotating UE (indexed by $k\!=\!5$) in one realization of scheduling period are illustrated in Fig. \ref{fig:scenario3}. The $5$-th UE cannot find appropriate combiner due to UE rotation and limited FoV in the middle of the scheduling period, leading to weak SNRs. It can be observed that the proposed scheme can predict the low SNR period of rotating UEs and schedule more transmission opportunities to them before the low SNR period, so that the packet drop rate can be significantly reduced. As a comparison, the benchmarks suffer from high packet drop rate during the low SNR period. This demonstrates the performance gain of the sensing-based channel prediction and the proposed predictive scheduling framework.

\begin{figure}[tb]
    \centering
    \includegraphics[width=0.8\linewidth]{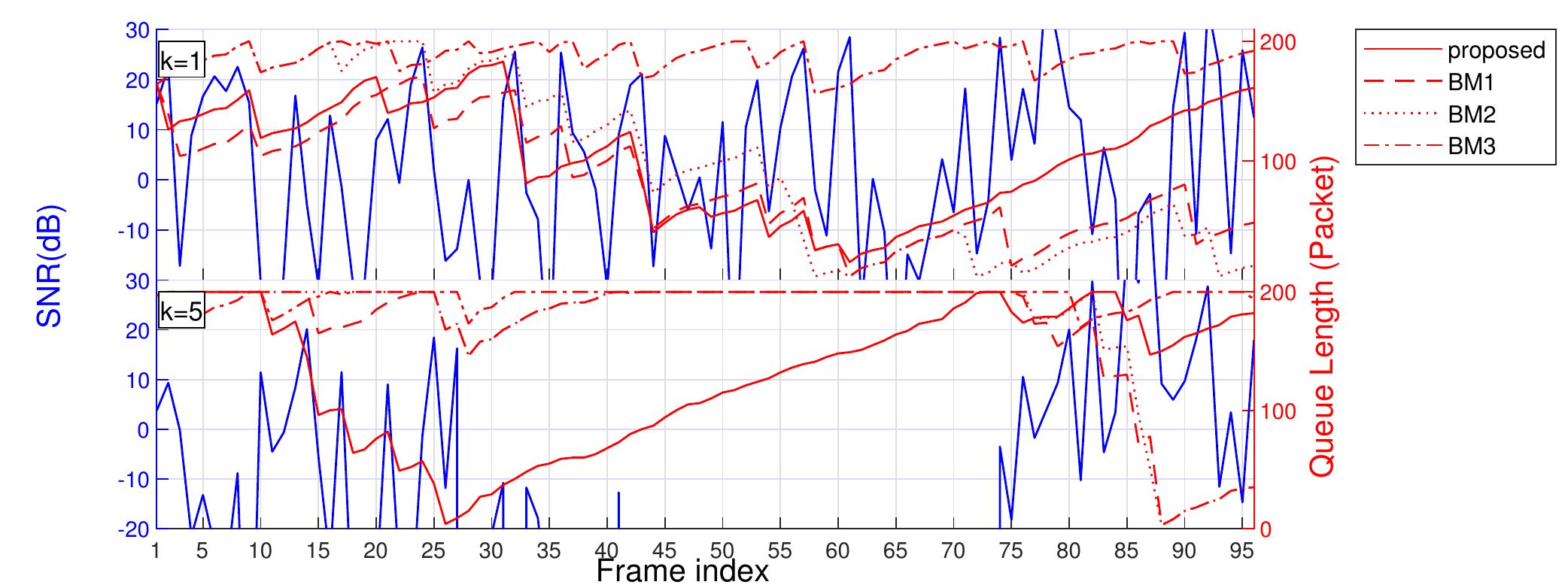}
    \caption{The dynamics of SNR and queue length. $T\!=\!100$, $Q_{\mathrm{max}}\!=\!200$, $Q_{1,k}\!\sim\!\mathcal{U}(100,200)$, $\lambda_{k}\!\sim\!\mathcal{U}(2,6)$, $W\!=\!400$ MHz, $N_{\mathrm{F}}\!=\!10$ ms, $B\!=\!30000$ bytes, $N_{\mathrm{R}}\!=\!16$, $N_{\mathrm{T}}\!=\!32$,  $\theta_{\mathrm{max}}\!=\!-\theta_{\mathrm{min}}\!=\!90^{\circ}$, $\phi_{\mathrm{max}}\!=\!-\phi_{\mathrm{min}}\!=\!30^{\circ}$, $P_{\mathrm{max}}\!=\!30$ dBm, $N_{k}^{\mathrm{cl}}\!=\!3$, $N_{k,i}^{\mathrm{ray}}\!=\!20$, $\sigma_{\alpha;k,i}^{2}\!\sim\!\mathcal{U}(4\!\times\!10^{-15},4\!\times\!10^{-14})$, $\sigma_{\phi;k,i}\!=\!\sigma_{\theta;k,i}\!=\!5^{\circ}$, $N_{0}\!=\!-174$ dBm/Hz, $w_{\mathrm{P}}\!=\!1500$, $w_{\mathrm{Q}}\!=\!2000$, and $P_{\mathrm{BSL}}\!=\!P_{\mathrm{BM}}\!=\!27$ dBm.}
    \label{fig:scenario3}
    \vspace{-0.5cm}
\end{figure}

\begin{figure}
    \centering
    \subfloat[]{
        \includegraphics[width=0.6\linewidth]{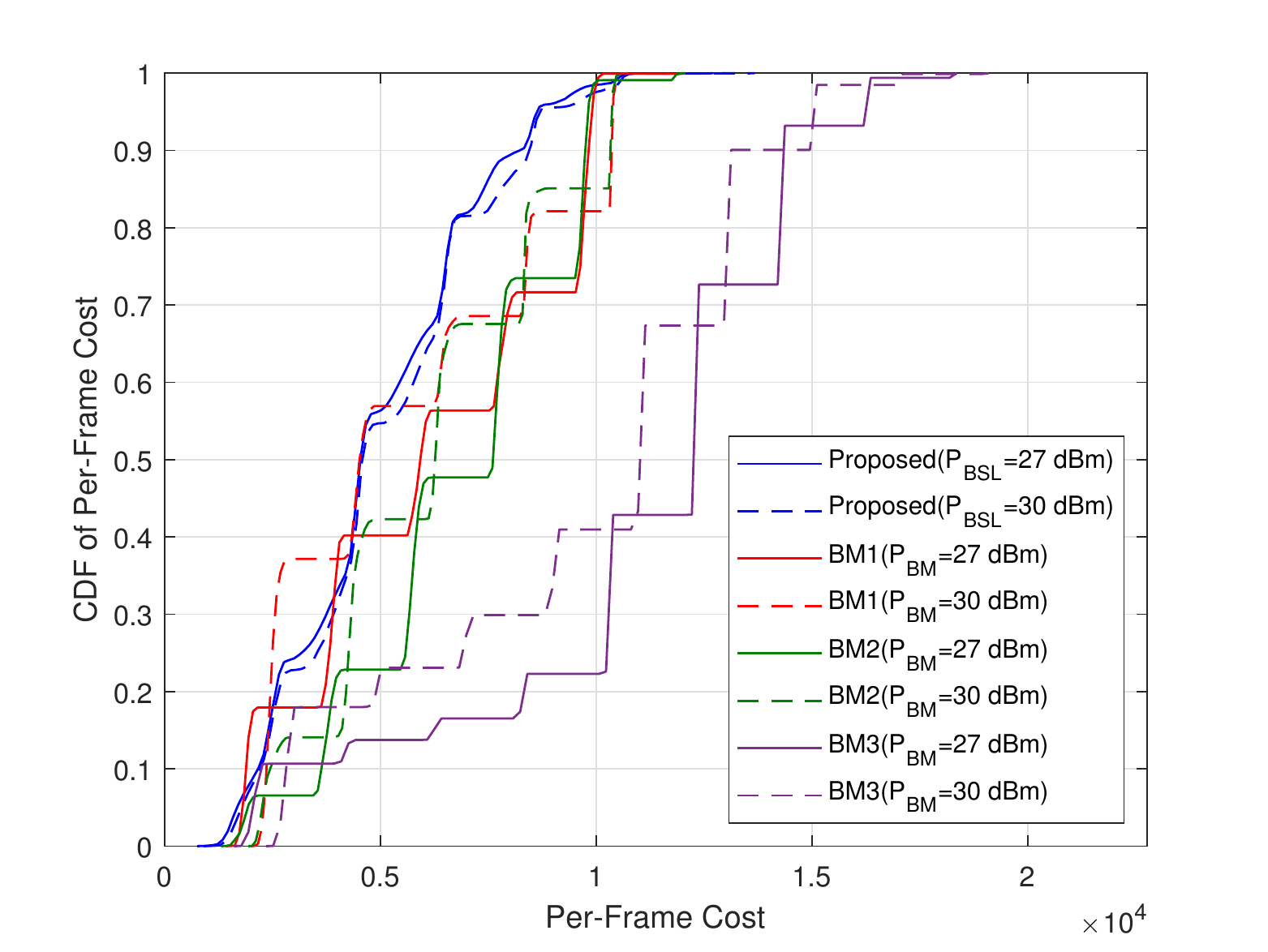}
    }
    \vspace{-0.3cm}
    \hfill
    \subfloat[]{
        \includegraphics[width=0.6\linewidth]{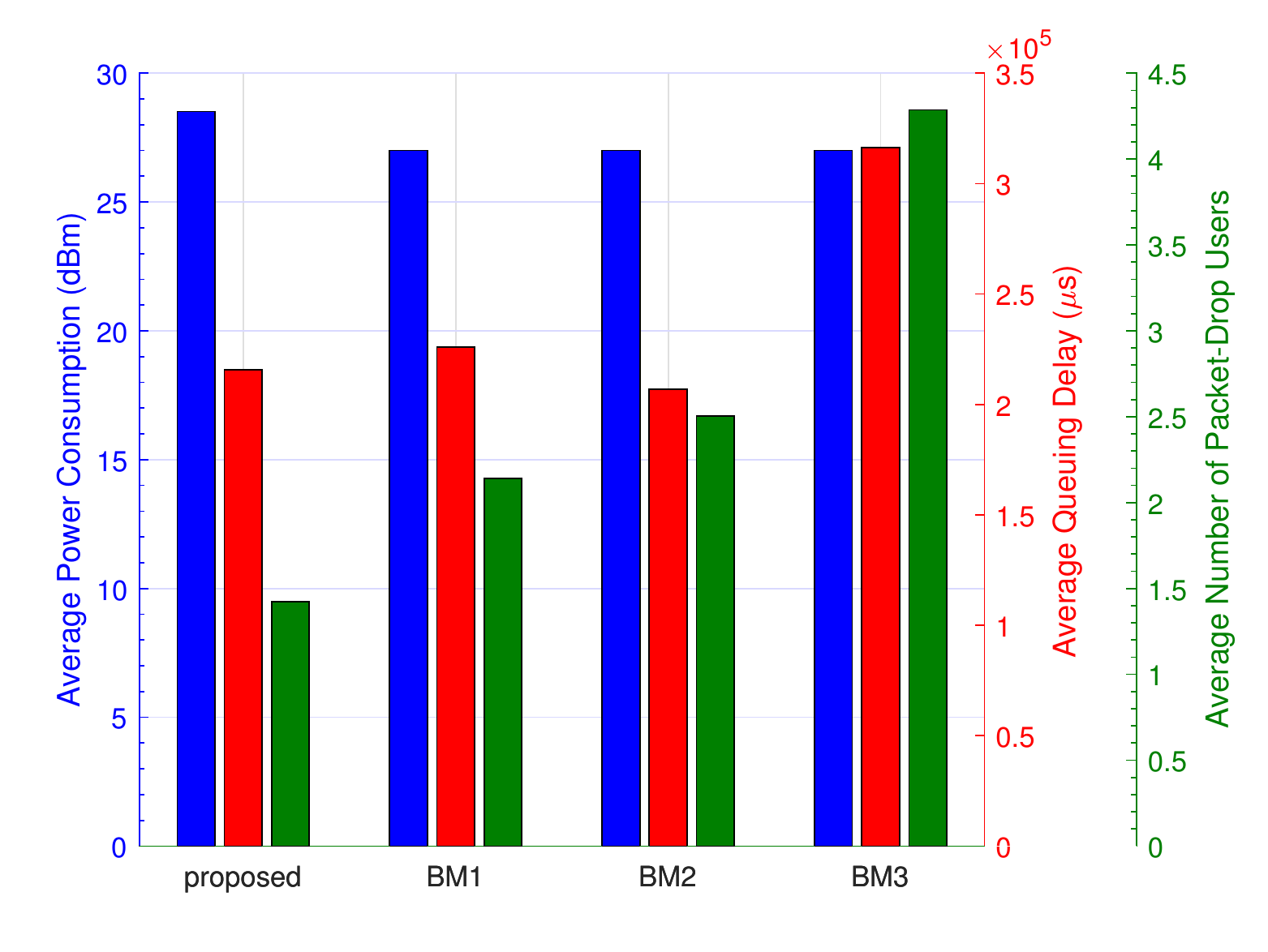}
    }
    \caption{(a) CDF of the per-frame cost. (b) Average transmission power, queuing delay and the number of packet-drop UEs.}
    \label{fig:CDF_and_bar}
    \vspace{-0.7cm}
\end{figure}

While Fig. \ref{fig:scenario3} shows the system performance in a snapshot, Fig. \ref{fig:CDF_and_bar}(a) shows CDFs of per-frame cost of the proposed scheme as well as the benchmarks. It can be observed that the proposed algorithm has significantly better CDF curves than the benchmarks. More insights can be obtained from Fig. \ref{fig:CDF_and_bar}(b), where average transmission power, queuing delay, and the number of packet drop UEs are illustrated. BM1 has the medium cost of delay and packet drop penalty. This is because it makes UE selection according to both the queue length and data rate. BM2 has the least cost of delay because it attempts to decrease the queue length as much as possible in every single frame. BM3 takes only the QSI into account but neglects the channel state information (CSI), which results in the worst performance. The proposed scheme manages to achieve the minimum packet drop rate, while keeping the average queuing delay in a low level. This demonstrates the benefits of channel prediction of the proposed scheme in suppressing the packet drop rate with non-stationary mmWave channel statistics.

\section{Conclusion}
\label{sec:conclusion}
In this paper, we consider the downlink transmission scheduling in an mmWave cell. Each UE is rotating with a predictable angular velocity for a number of frames, where the angular velocity of rotation can be measured by embedded motion sensors and reported to the BS. We first propose an SCSI-based beam alignment scheme. Then, we formulate the joint optimization of the downlink UE selection and power allocation as a finite-horizon MDP. To address the {\textit{curse of dimensionality}}, we finally propose a novel approximate MDP approach via one-step policy improvement over a base policy. Simulations show that the proposed MDP solution framework can effectively exploit the motion sensors to predict the future performance, resulting in an efficient scheduling algorithm.

\appendices

\section{PROOF of LEMMA \ref{lem:CDF_Y}}
\label{proof:CDF_of_Y}
Note that $
    Y_{t,k}
    \!=\!
    |\mathbf{w}_{t,k}^{\mathsf{H}}\mathbf{H}_{t,k}\mathbf{f}_{t,k}|^{2}
    \!=\!
    \Re^{2}(\mathbf{w}_{t,k}^{\mathsf{H}}\mathbf{H}_{t,k}\mathbf{f}_{t,k})
    \!+\!
    \Im^{2}(\mathbf{w}_{t,k}^{\mathsf{H}}\mathbf{H}_{t,k}\mathbf{f}_{t,k})
$. According to \cite{1033686}, when $N_{\mathrm{R}}$ and $N_{\mathrm{T}}$ are both sufficiently large, we have $
    \big|\mathbf{a}^{\mathsf{H}}_{\mathrm{R}}(\phi_{q_{t}^{\dagger}})\mathbf{a}_{\mathrm{R}}(\phi_{t,k,i,\ell})\big|^{2}
    \!\to\!
    \{0,1\}
$, $
    \big|\mathbf{a}^{\mathsf{H}}_{\mathrm{T}}(\theta_{p_{t}^{\dagger}})\mathbf{a}_{\mathrm{T}}(\theta_{t,k,i,\ell})\big|^{2}
    \!\to\!
    \{0,1\}
$. Denote the set $
    \mathcal{I}_{t,k,i}
    \!\triangleq\!
    \big\{\ell\big|
    |\mathbf{a}^{\mathsf{H}}_{\mathrm{R}}(\phi_{q_{t}^{\dagger}})\mathbf{a}_{\mathrm{R}}(\phi_{t,k,i,\ell})|^{2}
    \!\to\!
    1
$, $
    |\mathbf{a}^{\mathsf{H}}_{\mathrm{T}}(\theta_{p_{t}^{\dagger}})\mathbf{a}_{\mathrm{T}}(\theta_{t,k,i,\ell})|^{2}
    \!\to\!
    1
$, $
    \Lambda_{\mathrm{R}}(\phi_{t,k,i,\ell})\Lambda_{\mathrm{T}}(\theta_{t,k,i,\ell})
    \!=\!
    1
    \big\}
$, then $
    \Phi_{t,k,i}
    \!\triangleq\!
    |\mathcal{I}_{t,k,i}|
$ will follow the binomial distribution. Conditioned on $\{\Phi_{t,k,i}|\forall i\}$, the real and imaginary parts of $\mathbf{w}_{t,k}^{\mathsf{H}}\mathbf{H}_{t,k}\mathbf{f}_{t,k}$ will follow normal distributions, e.g., $
    \Re\big(\mathbf{w}_{t,k}^{\mathsf{H}}\mathbf{H}_{t,k}\mathbf{f}_{t,k}\big)\big|_{\{\Phi_{t,k,i}|\forall i\}}
    \!\sim\!
    \mathcal{N}\big(0,\frac{1}{2}\sum_{i\!=\!1}^{N_{k}^{\mathrm{cl}}}\Phi_{t,k,i}\sigma_{\alpha;k,i}^{2}\big)
$.
Remind that if $\chi_{2}^{2}$ is a random variable following chi-squared distribution with degrees of freedom $2$, then the CDF of $\chi_{2}^{2}$ is $
    F_{\chi_{2}^{2}}(x)
    \!=\!
    1
    \!-\!
    \exp(\!-\!x/2),\ x\!>\!0
$. Hence, Lemma \ref{lem:CDF_Y} is straightforward.

\section{PROOF OF LEMMA \ref{lem:analytical_expression}}
\label{proof:analytical_expression}
$\mathbf{u}_{t,k}$ and $\mathbf{s}_{t,\tau,k}$ represent the post-decision and pre-decision probability vector for the $k$-th queue respectively. $[\mathbf{c}^{(\!1\!)}]_{i}$ and $[\mathbf{c}^{(\!2\!)}]_{i}$ represent the per-frame queuing and packet-drop cost for the $k$-th UE in the $\tau$-th frame for cases $t\!+\!1\!\leq\!\tau\!\leq\!T$ and $\tau\!=\!T\!+\!1$, respectively. $\mathbf{M}_{t,k}$ and $\mathbf{P}_{k}$ are transition probability matrices for the $k$-th queue considering both packet departure and arrival and only the packet arrivals, respectively. In \eqref{eqn:local-value-analytical}, $
    \mathbf{s}_{t,\tau,k}^{\mathsf{T}}\mathbf{c}^{(\!1\!)}
$ and $
    \mathbf{s}_{t,T+1,k}^{\mathsf{T}}\mathbf{c}^{(\!2\!)}
$ counts for the average queuing and packet-drop cost in the $\tau$-th frame for cases $t\!+\!1\!\leq\!\tau\!\leq\!T$ and $\tau\!=\!T\!+\!1$, respectively. Hence, Lemma \ref{lem:analytical_expression} is straightforward.

\bibliographystyle{IEEEtran}
\bibliography{IEEEabrv,mmWave_Rotation_v1.0}

\end{document}